\documentclass[a4paper]{amsart}
\usepackage{amsmath,amssymb,amsfonts,amsthm}
\usepackage{latexsym}
\usepackage[dvips]{graphicx}
\usepackage{color}
\usepackage{indentfirst} 
\usepackage[footnotesize]{subfigure}
\usepackage{hhline}
\usepackage{booktabs}

\newtheorem{thm}{Theorem}

\newtheorem{rem}{Remark}
\newtheorem{defn}{Definition}

\newcommand{\R}{\mathbf{R}}
\newcommand{\bE}{\mathbb{E}}

\newcommand{\KL}{{K\!L}}

\begin{document}

\title{Partitioning a macroscopic system into independent subsystems}

\author{Luigi Delle Site}
\address{Institute for Mathematics, Freie Universit\"{a}t Berlin, D-14195 Berlin, Germany}
\email{luigi.dellesite@fu-berlin.de}

\author{Giovanni Ciccotti}
\address{Instituto per le Applicazioni del Calcolo, Consiglio Nazionale delle Ricerche, Rome Italy; Universit\`a La Sapienza, Rome, Italy; and School of Physics, University College Dublin, Belfield, Dublin 4, Ireland}      
\email{giovanni.ciccotti@roma1.infn.it}

\author{Carsten Hartmann}
\address{Institute of Mathematics, Brandeburgische Technische Universit\"at (BTU) Cottbus-Senftenberg, D-03046 Cottbus, Germany}
\email{carsten.hartmann@b-tu.de}

\begin{abstract}
We discuss the problem of partitioning a macroscopic system into a collection of independent  subsystems. The partitioning of a system into replica-like subsystems is nowadays a subject of major interest in several field of theoretical and applied physics, and the thermodynamic approach currently favoured by practitioners is based on a phenomenological definition of an interface energy associated with the partition, due to a lack of easily computable expressions for a microscopic (i.e.~particle-based) interface energy. In this article, we outline a general approach to derive sharp and computable bounds for the interface free energy in terms of microscopic statistical quantities. We discuss potential applications in nanothermodynamics and outline possible future directions.
\end{abstract}

%{37M25,82B05,82B24,94A17}

\maketitle
\section{Introduction}
We consider a system of $N$ interacting particles contained in a volume $V$ at temperature $T$ where the interactions in a general phase (or state) space $\mathcal{X}\subset\R^d$ are characterised by a Hamiltonian $H$. Its canonical partition function 
\begin{equation}
Z(N,V,T)=\int_{V\subset\mathcal{X}}e^{-\beta H(x)}dx
\end{equation} 
encodes the statistical and thermodynamic properties of the system. In this work we discuss the possibility of partitioning a macroscopic system into weakly interacting subsystems. Generally, the partitioning of a macroscopic system into several independent or weakly interacting subsystems, in particular, subsystems of the same size (replicas) is relevant for the simulation of various kinds of thermodynamic systems \cite{ED05}. Molecular modelling and simulation, for example, is often concerned with the question of how to quantify the trade-off between available computational resources and the desired spatial resolution of a model; running molecular simulations that allow for reliable and robust predictions often requires suitable techniques for the extrapolation of simulation data and estimated statistical quantities on macroscopic scales (see e.g.~\cite{raff}). The extrapolation to larger scales requires the knowledge of the error of thermodynamic quantities invoked by the passage from a small system to a larger one, which often boils down to quantifying the error invoked by approximating a large system by independent subsystems, each of which has the size of the small system of interest. To solve the latter problem, it is necessary to determine the energy associated with creating interfaces between the subsystems. Roughly speaking, if such an interface energy is sizable, compared to the total energy of each subsystem, then the system is not separable and, as a consequence, its thermodynamic properties are not scalable, otherwise the system is separable and thus the extrapolation to a larger scale is straightforward.

Beyond molecular simulation, the idea of partitioning a system into independent subsystems is known to be a powerful theoretical tool for proving the existence of thermodynamic limits in particle systems \cite{ll1,ll2}. 
A related conceptual problem, which is relevant in various disciplines and applications, is the definition of the size of a nanosystem, specifically, the question of how small a nanosystem must be so that one can manipulate it at the level of single molecules or atoms, while being large enough so that thermodynamic and statistical properties are still well defined (see e.g.~\cite{temprl,nanothe,jarz17}).

The idea of building an ensemble of independent replicas whose joint statistics mimics a large system has been used by T.L.~Hill to lay the foundations for what is today known as ``nanothermodynamics'' \cite{hillpap,hillbook,nanother2}. A closely related framework has been proposed in connection with the thermodynamics of non-additive systems \cite{ruffo2,latella,latella2,ruffo1,mori0}. 

Nowadays, the partitioning of thermodynamic systems is mostly done following Hill's phenomenological approach that is based on empirical laws and {\em ad hoc} assumptions about the underlying thermodynamics. 
Despite its lack of theoretical foundation, Hill's approach represents a powerful tool for understanding small systems thermodynamics, and the aim of this work is to make a first step towards linking the key quantities in Hill's phenomenological approach, such as the free energy associated with the separation of a system into smaller units, with microscopic quantities that can be estimated from running computer simulations. In doing so, we start from the microscopic definition of the canonical partition function of the system and derive an upper and lower bound for the free energy difference between the fully coupled system and the collection of replicas with the coupling removed.  We derive tight bounds for the interface free energy in terms of particle-based expressions, where it turns out that the precise form of the removed coupling terms is irrelevant for the analysis, so the considerations equally apply for collections of uncoupled or weakly coupled replicas. The fact that the bounds are particle-based has important practical consequences, in that it allows for a straightforward numerical implementation by molecular simulation algorithms.

\section{Hill's approach to replica-like systems}

Before we explain our approach, we briefly review Hill's replica model \cite{hillpap,hillbook}. 
The approach proposed by T.L.~Hill has been often taken as a paradigm for the construction of a macroscopic ensemble from statistically independent microscopic replicas (see e.g.~\cite{ruffo1, latella2}). Hill's empirical approach is targeted at systems that are so small so that their interaction with the environment is no more negligible, as a consequence of which the known laws of thermodynamics do not hold. The solution proposed by Hill is to consider the thermodynamics of a single nanosystem in an extended ensemble composed of non-interacting replicas, which represents a macroscopic system that is governed by the standard laws of thermodynamics. Here ``non-interacting'' refers to the absence of an explicit interaction potential between the particles in different replicas, but, as we will see below, the construction gives rise to an effective interaction on the level of the thermodynamic potentials.

If a macroscopic system is characterized by energy $E_{r}$, entropy, $S_{r}$, volume, $V_{r}$, and number of particles, $N_{r}$, then such a macroscopic system can be represented as a collection of $r$ independent or weakly interacting replicas. Letting $E,S,V,N$ denote the corresponding quantities for each of the replicas, then $E_{r}=r E$, $S_{r}=r S$, $V_{r}=r V$ and $N_{r}=r N$, and the fundamental equation reads 
\begin{equation}
dE_{r}=TdS_{r}-PdV_{r}+\rho dN_{r}+\mathfrak{X} dr\,,
\label{hill}
\end{equation}
where $T$, $P$ and $\rho$ are the temperature, pressure and chemical potential that are well defined for the macroscopic system. The key point of Hill' approach now is to introduce the unknown $\mathfrak{X}$, which is the energy per replica associated with splitting the system into subsystems, which plays the role of an effective interaction between replicas. This energy is only empirically defined, and most models following Hill's prescription tacitly assume that each replica can be characterized by a microscopic distribution (e.g.~canonical) so that intensive thermodynamic quantities in (\ref{hill}), such as temperature or chemical potential are also microscopically well defined. 

Despite the success of this approach, a problem remains: in a system of interacting particles one wishes to have a particle-based definition of $\mathfrak{X}$, according to the fundamental laws of statistical mechanics. In most applications, however,   $\mathfrak{X}$ is defined only empirically rather than being based on first principles. A rigorous way to compute  $\mathfrak{X}$ would require to have access to the difference $\Delta F= F-F_{0}$ between the free energy $F$ of the total system and the total free energy $F_0=F_1+F_2+\ldots$ of the subsystems, so as to link the properties of the macroscopic system and its non-interacting constituents. 

In the next section we will derive microscopic expressions that bound $\Delta F$ from above and from below, and discuss conditions under which these bounds are sharp. We refrain from presenting our results with maximum possible generality, so as to keep the arguments clear and concise. We stress, however, that most of the  following reasoning can be easily generalized at the expense of introducing a more elaborate mathematical notation.

\section{Bounds for the interface free energy}
Let
\begin{equation}\label{H}
H = H_0 + U
\end{equation}
where $H_0,U\colon {\mathcal X}\to\R$, $\mathcal{X}\subset\R^d$ are continuous functions that are sufficiently growing at infinity (coercive) and bounded below. Here 
\begin{equation}
H_0=H_1+H_2+\ldots
\end{equation} 
has the interpretation of the total energy of a finite collection  of uncoupled (or weakly coupled) subsystems, and $U$ is the coupling energy.  We set $\beta=1$ in what follows, and assume, without loss of generality, that $H_0$ and $U$ are dimensionless, which can be easily obtained by scaling $(H_0,U)\mapsto(\beta H_0,\beta U)$. The specific form of $H_0$ or $U$ is irrelevant for the following considerations. Specifically, we seek upper and lower bounds for the free energy 
\begin{equation}\label{DeltaF}
\Delta F = - \log\frac{Z}{Z_0}\,,
\end{equation}
with 
\begin{equation}\label{partitionFtcs}
Z = \int_\mathcal{X} e^{-H(x)} dx\,,\quad Z_0 = \int_\mathcal{X} e^{-H_0(x)}dx
\end{equation}
We assume that both $Z$ and $Z_0$ are finite, so that $\Delta F$ is well-defined. 
It will be convenient to consider a dual representation of the thermodynamic free energy that is based on relative entropy or Kullback-Leibler divergence. 

\begin{defn}\label{defn:kl}
	Let $f,g\colon\mathcal{X}\to[0,\infty)$ be two probability density functions on ${\mathcal X}$. The Kullback-Leibler divergence (or: relative entropy) between $f$ and $g$ is defined as 
	\begin{equation}\label{kl}
	\KL(f,g) = \int_\mathcal{X}\log\frac{f(x)}{g(x)}f(x)\,dx
	\end{equation}
	provided that
	\begin{equation}
	\int_{\{g(x)=0\}} f(x)dx = 0\,,
	\end{equation}
	otherwise we set $\KL(f,g)=\infty$. 
\end{defn}

\begin{rem}
	The above definition is based on the convention $0\log 0 = 0$, which can be justified by a limit argument.
	The requirement that division by zero is only allowed on null sets, i.e. $\{g(x)>0\}\subset\mathcal{X}$ is a null set under the probability measure
	\begin{equation}
	P(A) = \int_A f(x)dx\,,\quad A\subset\mathcal{X} \textrm{ measurable}
	\end{equation}
	induced by the density function $f$, expresses the absolute continuity of the probability measure $P$ with respect to the measure $Q$ induced by $g$, so that $Q(B) > 0$ whenever $P(B)>0$ for a measurable set $B\subset\mathcal{X}$.  
	Absolute continuity guarantees that the integral (\ref{kl}) is well defined, possibly having the value $+\infty$. 
	
%	In the following we will consider only strictly positive probability densities. 
\end{rem}

 	Note that the Kullback-Leibler (KL) divergence is not symmetric in its arguments, i.e.~$\KL(f,g)\neq \KL(g,f)$. A useful property, however, is that it is non-negative, 
	\begin{equation}\label{nonnegKL}
	\KL(f,g)\ge 0\,,
	\end{equation}
	with equality if and only if $f(x)=g(x)$ for all $x\in\mathcal{X}$ up to a set $N\subset\mathcal{X}$ of measure zero under the probability density $f$. To see this, notice that, by Jensen's inequality, 
	\begin{equation}
	\begin{aligned}
	 \int_\mathcal{X}\log\frac{f(x)}{g(x)}f(x)\,dx & =  -\int_\mathcal{X}\log\frac{g(x)}{f(x)}f(x)\,dx\\
	  & \ge  -\log\int_\mathcal{X}\frac{g(x)}{f(x)}f(x)\,dx\\
	  & = -\log\int_\mathcal{X} g(x)dx\\
	  & = 0\,,
	\end{aligned}
	\end{equation}
	where the last identity is implied by $g$ being a probability density. As the logarithm is \emph{strictly} concave, equality in Jensens inequality is attained if and only if the integrand $\log(g/f)$ in the first equation is constant up to a set of measure zero under the probability density $f$. Since $f\ge 0$ cannot be identically zero, the integral must be equal to zero to attain equality, therefore $g=f$ up to null sets. 
	
	In the following we assume that the first argument in the KL divergence (here: $f$) is always strictly positive, which implies that null sets under $f$ are Lebesgue null sets, i.e., the above statements hold almost everywhere (a.e.). 

\subsection{Two-sided Bogoliubov inequality}

Finding an upper and a lower bound for the free energy difference $\Delta F=-\log(Z/Z_0)$ is now straightforward. To this end we define the two probability densities 
\begin{equation}\label{densities}
\pi = \frac{\exp(-H)}{Z}\,,\quad \pi_0 = \frac{\exp(-H_0)}{Z_0}
\end{equation}
on ${\mathcal X}$, 
which are readily seen to be normalized. 

\begin{thm}[Two-sided Bogoliubov inequality]\label{thm:2BI} Under the previous assumptions, it holds that 
	\begin{equation}\label{2BI}
	\bE_\pi[U] \le \Delta F \le  \bE_{\pi_0}[U]\,,
	\end{equation}
where 
\begin{equation}
\bE_{\pi}[U] = \frac{1}{Z}\int_\mathcal{X} U(x)\pi(x)\,dx\,,\quad \bE_{\pi_0}[U] = \frac{1}{Z_0}\int_\mathcal{X} U(x)\pi_0(x)\,dx
\end{equation}
denote the expectations with respect to $\pi$ and $\pi_0$.
\end{thm}
\begin{proof}
	The conditions on $H_0$ and $U$ imply that $U$ is integrable with respect $\pi_0$ and, as a consequence, it is integrable with respect to $\pi$ as well. Moreover both $\pi_0$ and $\pi$ are strictly positive. The non-negativity of the KL divergence (see Definition \ref{defn:kl}) then implies that 
	\begin{equation}
	0 \le \KL(\pi,\pi_0) = -\log\frac{Z}{Z_0} - \bE_\pi[U]\,,
	\end{equation}
	in other words,
	\begin{equation}
	\Delta F \ge \bE_\pi[U]\,.
	\end{equation}
	Interchanging the arguments $\pi$ and $\pi_0$ in the KL divergence, we see that 
	\begin{equation}
	0 \le KL(\pi_0,\pi) = \log\frac{Z}{Z_0} + \bE_{\pi_0}[U]\,,
	\end{equation}
	which gives us the upper bound
	\begin{equation}
	\Delta F \le \bE_{\pi_0}[U]\,.
	\end{equation}
	Hence the assertion is proved. 
\end{proof}

The upper bound $\Delta F \le \bE_{\pi_0}[U]$ of Theorem \ref{thm:2BI}, that follows directly from (\ref{DeltaF}) via Jensen's inequality, is known as Bogoliubov inequality (also: Peierls-Bogoliubov or Gibbs-Bogoliubov inequality) in the literature which is why we named the result of Theorem \ref{thm:2BI} after Bogoliubov; see \cite{peier,symanzik,ll1,ll2}. We mention a related result  by Isihara \cite{GBdagger} who has proved a two-sided Bogoliubov inequality using a cumulant expansion of the free energy. 

\subsection{Variational bounds}

By Jensen's inequality, the bounds in Theorem \ref{thm:2BI} are sharp if and only if $U$ is constant, in which case $\Delta F=U$. Excluding this somewhat pathological case, and assuming further that $U$ is bounded, we can still improve the bounds for $\Delta F$ by using the duality relation \cite{DMR96,Hetal14,HS12}
\begin{equation}\label{DV}
-\log \bE_{\pi_0}[e^{-U}] = \inf_{\rho\in\mathcal{P}_+} \left\{\bE_{\rho}[U] + \KL(\rho,\pi_0)\right\}\,,
\end{equation} 
where, by definition, $-\log \bE_{\pi_0}[e^{-U}]=\Delta F$
and the infimum is taken over the set $\mathcal{P}_+$ of a.e.~positive probability densities $\rho> 0$ on $\mathcal{X}$. Equation (\ref{DV}) is known as the Donsker-Varadhan principle \cite{DS}, which, in our case, is again a simple consequence of Jensen's inequality: Noting that, for any probability density $\rho>0$ (a.e.) and any bounded random variable $W$ on ${\mathcal X}$, we have 
\begin{equation}
\bE_{\pi_0}[W]  = \bE_{\rho}[WL]
\end{equation}
where the weight $L(x)=\pi_0(x)/\rho(x)$, $\rho(x)\neq 0$ is the likelihood ratio between the two probability densities $\pi_0$ and $\rho$, and 
\begin{equation}
\bE_{\rho}[W] = \int_{\mathcal{X}}W(x)\rho(x)dx
\end{equation}
denotes the expectation of $W$ with respect to $\rho$. (Without loss of generality, we can exclude the set $\{\rho(x)=0\}\subset{\mathcal X}$ from our considerations, which is a null set under both $\rho$ and $\pi_0$.) Hence, by Jensen's inequality, it follows that 
\begin{equation}\label{auxeq1}
-\log \bE_{\pi_0}[e^{-U}] 
 = -\log \bE_{\rho}\!\left[e^{-U+\log L}\right] \le \bE_{\rho}\!\left[U-\log L\right]
\end{equation}
The rightmost expression equals $\bE_{\rho}[U]+\KL(\rho,\pi_0)$.
By the strict convexity of the exponential function in (\ref{auxeq1}), equality can only be attained if $U(x)+\log(\pi_0(x)/\rho(x))$ is a.e.~constant, since otherwise the inequality becomes strict. Moreover, since the right hand side in (\ref{auxeq1}) is strictly convex in $\rho$, there is a unique probability density function $\rho=\rho^*$ for which the equality is attained. It is given by 
\begin{equation}
\rho^*(x) = \exp(\gamma-U(x))\pi_0(x)\,,
\end{equation}
and it is therefore the unique, positive minimizer of (\ref{DV}). Here $\gamma>0$ is the normalization constant, and it can be readily verified that $\gamma$ is equal to $\Delta F$ which implies that $\rho^*=\pi$ and thus yields the known equality
\begin{equation}\label{infDV}
\Delta F = \bE_{\pi}[U] + \KL(\pi,\pi_0)\,
\end{equation}
from the first step in the proof of Theorem \ref{thm:2BI}.

In practice, the rightmost term in (\ref{infDV}) will be difficult to evaluate, however we can get a good estimate by exploiting the fact that the Donsker-Varadhan principle (\ref{DV}) has a dual form that characterises the (relative) entropy in terms of the free energy and the internal energy (see \cite[Prop.~2.3]{DMR96}):  
\begin{equation}\label{dualDV}
\KL(\pi,\pi_0) = \sup_{\psi\in L^1}\left\{\bE_{\pi}[\psi] - \log\bE_{\pi_0}[e^\psi]\right\}\,, 
\end{equation}
where integrability $\psi\in L^1$ is understood with respect to $\pi$. 
By combining the equations (\ref{DV})--(\ref{dualDV}), we obtain the following variational characterisation of $\Delta F$. 
\begin{thm}[Variational bounds]\label{thm:varBds}
	We have 
	\begin{equation}\label{2BItight}
\sup_{\psi\in L^1} \left\{\bE_{\pi}[U+\psi] - \log\bE_{\pi_0}[e^\psi]\right\} = \Delta F =  \inf_{\rho\in\mathcal{P}_+} \left\{\bE_{\rho}[U] + \KL(\rho,\pi_0)\right\}\,.
	\end{equation}
That is, for all $\psi\in L^1$ and $\rho\in\mathcal{P}_+$ it holds that 
\begin{equation}\label{2BI2}
\bE_{\pi}[U+\psi] - \log\bE_{\pi_0}[e^\psi] \le \Delta F  \le \bE_{\rho}[U] + \KL(\rho,\pi_0)\,.
\end{equation}
\end{thm}

Even though the variational form (\ref{2BItight}) may not be particularly useful in practice, (\ref{2BI2}) can be the starting point for systematic improvements of the Bogoliubov bound  (\ref{2BI}). In particular (\ref{2BI}) is obtained from  (\ref{2BI2}) by setting $\psi=0$ and $\rho=\pi_0$, and so a sensible strategy to improve (\ref{2BI}) could be to consider the homotopies
\begin{equation}
\psi_{\alpha_1} = -\alpha_1 U\,,\quad 0\le \alpha_1\le 1\,,
\end{equation}
and 
\begin{equation}
\rho_{\alpha_2} = \exp(\gamma(\alpha_2) - \alpha_2 U)\pi_0\,,\quad 0\le \alpha_2\le 1\,,
\end{equation}
in (\ref{2BI2}) that have the property that one recovers the original bound (\ref{2BI}) for $\alpha_1=\alpha_2=0$ and the sharp bound in (\ref{2BItight}) for $\alpha_1=\alpha_2=1$. Here $\gamma(\alpha_2)$ is again the normalization constant that guarantees that $\rho_{\alpha_2}$ is a probability density.

By differentiating the variational bounds with respect to the homotopy parameters $\alpha_1$ and $\alpha_2$, it can be readily seen that the lower bound in (\ref{2BI2}) is strictly increasing for small values of $\alpha_1$, whereas the upper bound is strictly decreasing for small values of $\alpha_2$, in other words, small perturbations to $\psi=0$ and $\rho=\pi_0$ are sufficient to improve the bounds. 

\begin{rem}
	
	Equation (\ref{DV}) furnishes the Gibbs relation $F=E-TS$ for the
	thermodynamic free energy $F$, with $E$ being the internal energy, $T$ temperature and $S$ denoting the Gibbs entropy. We specify the previous assumptions by assuming that ${\mathcal X}\subset\R^d$ is compact and setting $H_0\equiv 0$, so that $\pi_0\equiv 1/|{\mathcal X}|$. Then, setting $U=\beta V$ with $\beta=(k_B T)^{-1}$, $k_B>0$ being Boltzmann's constant and $V$ denoting any continuous energy function on $\mathcal{X}$, equation (\ref{DV}) turns into 
	\begin{equation}
	\underbrace{-\beta^{-1} \log \int \exp(-\beta V)\,dx}_{=F} = \min_{\rho \in\mathcal{P}_+}\left\{\underbrace{\int V \rho \,dx}_{=E} +  \underbrace{\beta^{-1}\int\rho\log\rho\,dx}_{=-TS} \right\}\,.
	\end{equation}
	The unique minimizer is the Gibbs-Boltzmann density
	$\rho^*=\exp(-\beta V)/Z$ with normalization constant $Z=\exp(-\beta F)$. 
	Hence the Donsker-Varadhan principle can be considered a measure-theoretic generalization of the Gibbs relation for the free energy in terms of internal energy and (relative) entropy. 
\end{rem}

\subsection{Examples and computational aspects}

As a simple example, we consider a collection of two subsystems with separable Hamiltonian
\begin{equation}
H_0\colon\mathcal{X}_1\times\mathcal{X}_2\to\R\,,\; (x_1,x_2)\mapsto H_1(x_1)+H_2(x_2)
\end{equation}
with $H_i\colon\mathcal{X}_i\to\R$ and $\mathcal{X}_i\subset\R^{6N_i}$, $i=1,2$. The coupling is by virtue of the interaction Hamiltonian
\begin{equation}
U\colon\mathcal{X}_1\times\mathcal{X}_2\to\R\,,\; (x_1,x_2)\mapsto U(x_1,x_2)\,.
\end{equation}
Letting 
\begin{equation}
Z_i(N_i,V_i,T) = \int_{V_i\subset\mathcal{X}_i} e^{-\beta H_i(x_i)}dx_i\,,\; i=1,2\,,
\end{equation}
and 
\begin{equation}
Z(N,V,T) = \int_{V\subset\mathcal{X}} e^{-\beta H(x)}dx\,,
\end{equation}
with $\beta=1/(k_BT)$ be the canonical partition functions, where we remind the reader of the convention $H=H_0+U$ with $H_0=H_1+H_2$, the free energy difference is defined as 
\begin{equation}
\Delta F = F - F_1 - F_2 =  -\beta^{-1}\log\left(\frac{Z}{Z_1Z_2}\right).
\end{equation}
Here  $N=N_{1}+N_{2}$, $|V|=|V_{1}|+|V_{2}|$ are the total number of particles and the total volume, with constant density 
\begin{equation}
\varrho = \frac{N}{|V|}=\frac{N_{1}}{|V_{1}|}=\frac{N_{2}}{|V_{2}|}>0\,.
\end{equation}
Theorem \ref{thm:2BI} now implies that 
\begin{equation}
\bar{U} \le \Delta F\le \bar{U}_0
\end{equation} 
where 
\begin{equation}
\bar{U}= \frac{1}{Z}\int_{V}U(x)e^{-\beta H(x)}dx\,
\end{equation}
and
\begin{equation}
\bar{U}_0= \frac{1}{Z_0}\int_{V}U(x)e^{-\beta H_0(x)}dx\,,\quad Z_0=Z_1Z_2\,,
\end{equation}
denote the expected values of the interaction potential between the particles of subsystem 1 and subsystem 2 with (lower bound) or without coupling (upper bound).
Since the two subsystems are independent, $\bar{U}_{0}$ has the interpretation of an average interdomain energy: the energy required to partition the system into two subsystems plus a ``{\it corridor}'', assuming a phase space volume $\delta V$ that is negligible in comparison with $V_{1}$ and $V_{2}$, so that $|V|\approx |V_1|+|V_2|$.\\

Other examples that are relevant in applications are the calculation of thermodynamic bulk properties or the simulation of open systems \cite{njp,prejing}. Imagine, for example, a situation in which one wants to study the conformational change of a large biomolecule immersed in aqueous solution. Typically one would need a large system to properly represent the bulk of the solvent, however such calculations are very expensive, and one may think of separating the system into a solvation region and an environement, and simulating only the relevant subsystem consisting of the biomolecule and its solvation shell. 
If one were able to simulate only the subsystem while ignoring the interaction with the (possibly infinitely large) solvent environment, it would be possible to compute the relevant properties by simulating a system of a drastically reduced dimensionality. In order to quantify finite-size effects that are invoked by neglecting the interactions between the subsystem and its environment (mediated by the interaction potential $U$), it is crucial to have precise estimates of the interface free energy associated with the solvation shell.

Roughly speaking, finite-size effects are negligible, and computing bulk properties by simulating a small system is justified, when $\Delta F$ is small compared to the characteristic energy scale of the system, otherwise the effects coming from the environment have to be taken into account, for example by increasing the system size or by changing the interaction potential \cite{pawar}. As these free energy calculations are extremely costly and the simulations have to be repeated many times with different simulation parameters (e.g.~see \cite{noy}) in order to determine the optimal system size or the parameters of the interaction potential, it is important to have free energy estimators that are computationally feasible.

Let us stress that the aforementioned examples are just representatives for a much larger class of problems, for most molecular dynamics models that involve either heat baths, reservoirs, or the alike are subject to finite size effects when implemented on a computer. Assessing whether these finite size effects are important or not is crucial, and in this paper we suggest a flexible tool to derive quantitative estimates to see how far a simulation is away from the (ideal) thermodynamic limit, without running an enormous number of trial simulations.

Being able to split a large system with $d=3N$ degrees of freedom into many, say, $r$ independent subsystems moreover allows for easy parallelizability of MD codes \cite{griebel}. This typically reduces the computational complexity of the corresponding algorithms from $\mathcal{O}(d^\alpha)$, with $\alpha$ typically between 2 and 3, to $\mathcal{O}(d_1^\alpha+\ldots +d_r^\alpha)$ with $N_k=d_k/3$ being the number of particles in the $k$-th subsystem, even though one is simulating in total the same number of particles $N=N_1+\ldots+N_r$.

\begin{rem}
An often neglected aspect in free energy calculations concerns the corresponding statistical estimators. Because the free energy is of the form of a cumulant generating function $\log \bE[\exp(\ldots)]$, standard Monte Carlo estimators may suffer from both a large variance (because of the exponential function) and a systematic bias (because of the logarithm); see e.g.~\cite[Sec.~2.3--2.4]{FEbook}. One resort is to replace the Monte Carlo estimator by its variational counterpart from Theorem \ref{thm:varBds} and use stochastic optimization tools to approximate the free energy as has been described in  \cite{HS12,ZhangEtal14}, or to use suitable moment estimates (``concentration inequalities'') for the underlying probability distributions as in e.g.~\cite{bousquet,watbled}.
\end{rem}

\section{Discussion and Conclusions}

Mimicking infinite physical systems by simulating finite domains is challenging because the statistics of the small system is by definition ``marginal'', in that one considers finitely many degrees of freedom of a system with, theoretically, infinitely many degrees of freedom. Therefore one needs suitable numerical techniques to account for the finite size effects in the corresponding probability distributions \cite{previctor}.
Particularly important are error indicators for interface free energies that allow for estimating the statistical error due to the separation of an infinite system into weakly interacting small subsystems. The difference between the upper and lower bounds as well as their variational counterparts provide easily computable error bounds to quantify this effect.

Getting precise estimates of the interface free energy $\Delta F$ is of interest in many practical problems, in which the separation of a system into non-interacting or weakly interacting subsystems plays a role \cite{njp,prejing}; examples include the formation of islands in colloidal systems \cite{deb} or the formation of nanosize domains through epitaxial growth \cite{germsilic}, to mention just two important cases; specifically, in colloidal systems, the knowledge of the degree of independence of the various islands in terms of the ratio between the interface energy and the total energy of the islands gives a handy criterion as to whether it is justified to treat single islands as systems per se in a simulation and thus acquire a large gain in computational efficiency in analyzing the local microscopic details.
Given the current development of nanothermodynamics in the direction of (bio-)technology, precise and computationally feasible free energy bounds may be further relevant for the statistical modeling of such nanosystems \cite{nanoeng,hillnanotech,nanonature,jcpvirus}. The variational free energy bounds can moreover serve as quantitative criteria to optimally design separable systems with specific degrees of separation by suitably adjusting the characteristic physical parameters (e.g.~interaction range, strength of interaction, atom spacing). One such  example are molecular wires in nanoelectronic devices, for which a high degree of separability implies that the independent subsystems are representative for locally acting devices, where the degree of locality can be engineered by  a proper choice of the physical parameters controlled by the minimization of the interface energy \cite{molwir}. We stress one more time that the above considerations can be readily generalized to non-stationary and path-space problems (cf.~\cite{DKPP16}). Future work should address bounds for specific observables. 
%
%as a consequence our upper and lower bound estimate provide a general criterion of unambiguous separation of a system in subsystems. In fact the difference between the upper and lower bound of $\Delta F$ can be used as criterion to establish whether or not one can consider the system separable and within which error this approximation is valid. 

\section*{Acknowledgments}
 This work was supported by the Deutsche Forschungsgemeinschaft (DFG) under the grant CRC 1114 (projects A05 and C01) and by the European Community through the project E-CAM. 

\end{document}